\documentclass{siamart}

\usepackage{amsfonts}
\usepackage{graphicx}
\usepackage{epstopdf}
\usepackage{algorithmic}
\ifpdf
  \DeclareGraphicsExtensions{.eps,.pdf,.png,.jpg}
\else
  \DeclareGraphicsExtensions{.eps}
\fi

\newcommand{\TheTitle}{Form factor (Fourier shape transform) of polygon and polyhedron.}
\newcommand{\TheAuthors}{Joachim Wuttke}

\headers{Form factor of polygon and polyhedron.}{\TheAuthors}

\title{{\TheTitle}}

\author{
  Joachim Wuttke\thanks
 {J\"ulich Centre for Neutron Science (JCNS)
  at Heinz Maier-Leibnitz Zentrum (MLZ),
  Forschungszentrum J\"ulich GmbH,
  Lichtenbergstra\ss e 1, 85747 Garching, Germany
    (\email{j.wuttke@fz-juelich.de}, \url{http://www.fz-juelich.de/SharedDocs/Personen/JCNS/EN/Wuttke_J.html}).}
}

\usepackage{amsopn}

\ifpdf
\hypersetup{
  pdftitle={\TheTitle},
  pdfauthor={\TheAuthors}
}
\fi

\renewcommand{\v}[1]{\ensuremath{\mathbf{#1}}}
\renewcommand{\d}{\text{\rm d}}
\newcommand{\e}{\text{\rm e}}
\renewcommand{\r}{\v{r}}
\newcommand{\q}{\v{q}}
\newcommand{\n}{\v{\hat n}}
\newcommand{\V}{\v{V}}
\newcommand{\E}{\v{E}}
\newcommand{\R}{\v{R}}
\newcommand{\xabs}{v}
\newcommand{\x}{\v{\xabs}}
\newcommand{\pa}{\parallel}
\newcommand{\qpa}{\q_\pa}

\newcommand{\qrperp}{\q_\perp\r_\perp}

\newcommand{\qcross}{\q_\times}
\newcommand{\Plane}{\mathcal{E}}
\def\JJ{{\tilde J}}
\def\KK{{\tilde K}}
\DeclareMathOperator{\sinc}{sinc}
\DeclareMathOperator{\Ar}{Ar}
\DeclareMathOperator{\Vol}{Vol}
\newcommand{\Nabla}{\v{\nabla}}
\DeclareMathOperator{\Real}{Re}

\usepackage{stmaryrd}
\usepackage{mathtools}
\usepackage{amssymb}

\newsiamremark{remark}{Remark}

\begin{document}

\maketitle

\begin{abstract}
  The Fourier transform of the indicator function of arbitrary polygons and polyhedra is computed
  for complex wavevectors.
  Using the divergence theorem and Stokes' theorem,
  closed expressions are obtained.
  Apparent singularities, all removable, are discussed in detail.
  Loss of precision due to cancellation near the singularities
  can be avoided by using series expansions.
\end{abstract}

\begin{keywords}
  Polygon, polyhedron, form factor, shape transform, Fourier transform
\end{keywords}

\begin{AMS}
  42B10 51M20 65D20
\end{AMS}

\section{Introduction}
The term \textit{``form factor''} has different meanings in science and in engineering.
Here, we are concerned with the form factor of a geometric figure
as defined in the physical sciences,
namely the \textit{Fourier transform} of the figure's indicator function,
also called the \textit{shape transform} of the figure.

This form factor has important applications in the emission, detection,
and scattering of radiation.
Two-dimensional shape transforms
are used in the theory of reflector antennas~\cite{LeMi83}.
The shape transform of three-dimensional nanoparticles is used to interpret
neutron and x-ray small-angle scattering (SAS, SANS, SAXS)~\cite{Ham10}.
A particularly rich multitude of mostly polyhedral shapes is observed
for particles grown on a substrate \cite{Hen05},
as observed by grazing-incidence neutron and x-ray small-angle scattering
(GISAS, GISANS, GISAXS)~\cite{ReLL09}.
Thence it comes that possibly the most extensive collection
of particle shape transforms published to this date
can be found in the documentation \cite{ReLL09,Laz08}
of a GISAXS simulation software called IsGISAXS.
The oldest specimens are the form factors of a cylinder and of a sphere,
worked out by Lord Rayleigh in his electromagnetic theory of light scattering~\cite{Ray81}.

Analytic expressions for form factors of polyhedra
can be obtained by straightforward triple integration, using Fubini's theorem.
However, except for very simple shapes like a rectangular box, oriented along the coordinate axes,
the resulting expressions look more complicated than one might expect for such a plain problem;
they do not reflect the symmetry of the geometric figure,
and they contain removable singularities that cause two difficulties:
Division by zero at the singularities,
unless analytic continuations are used,
and loss of arithmetic precision due to cancellation of terms near the singularities,
unless special precautions are taken.

Deriving and implementing numerically stable algorithms for a considerable collection
of Platonic solids, prisms, truncated prisms, pyramids, frusta, and more would be an immense labor.
It is preferable to derive the shape transform of arbitrary polygons and polyhedra
in generic terms,
and deal once and for all with all possible singularities.
For polygons, a generic expression for the form factor has been derived decades ago~\cite{LeMi83};
singularities were not discussed.
For polyhedra, a first-order approximation to the form factor is used in a novel
GISAXS simulation software, called HipGISAXS \cite{ChSL13}.
Here, exact expressions for the form factors of polygons and polyhedra will be derived.
Proper treatement of all singularities will lead to practicable algorithms
for computing at arbitrary wavevectors.
All results presented here have already found an application in another novel
GISAXS simulation software, called BornAgain \cite{ffp:ba}.

For a Fourier transform in the usual sense,
and for most applications, wavevectors are real: $\q\in\mathbb{R}^3$.
In GISAS, however, the incident and scattered radiation may undergo substantial absorption,
which can be modeled by an imaginary part of~$\q$.
Similarly,  an imaginary part of~$\q$ can be used to described
light amplification in a laser medium.
Therefore, we admit complex wavevectors $\q\in\mathbb{C}^3$.

Polygonal and polyhedral form factors have strong similarities in their mathematical structure,
reflecting the analogy of Stokes' theorem and the divergence theorem.
This paper strives to bring out these similarities
by a purposedly parallel treatment of the two- and the three-dimensional case.
It is organized as follows:
\Cref{S2fig} provides some basics for computing the form factor of a planar figure
embedded in three-dimensional space.
\Cref{S2poly} deals with the form factor of a polygon;
it provides a closed analytic expression and power series coefficients.
\Cref{S3fig} treats the form factor of three-dimensional figures;
\cref{S3poly} specializes to a polyhedron.
\Cref{Si} works out simplifications for figures with inversion centers.
TODO \cref{Stests}

As for notation:
Vector products (dot products and cross products) are linear in both components;
complex conjugation, where needed, is denoted explicitely.
And of course $0^0=1$ \cite{Knu92}.

\section{Form factor of a planar figure}\label{S2fig}

\begin{definition}[oriented plane]\label{Dplane}
  A \emph{plane},
  given by a normal unit vector~$\n$
  and a signed distance~$r_\perp$ from the origin,
  shall be denoted as
  \begin{equation}
    \Plane(\n,r_\perp)
    \coloneqq \{\r\in\mathbb{R}^3 \;|\; \r\n=r_\perp\}.
  \end{equation}
  The \emph{orientation} of the plane is induced by~$\n$ as follows:
  If the ordered triple $(\v{b}_1,\v{b}_2,\n)$ of pairwise orthogonal vectors
  is a positively oriented (right-handed) base of the~$\mathbb{R}^3$,
  then the ordered pair $(\v{b}_1,\v{b}_2)$ is a positively oriented base of~$\Plane(\n,r_\perp)$.
\end{definition}

  For a loop in the plane~$\Plane(\n,r_\perp)$,
  the following statements are equivalent:
  (1) The loop has a positive winding number.
  (2) The loop runs counterclockwise.
  (3) The loop fulfills the right-hand-rule with respect to~$\n$.

\begin{definition}[form factor of a 2d figure in 3d space]\label{Dff23}
  The form factor of a two-dimensional planar figure~$\Gamma\subset\Plane(\n,r_\perp)$,
  embedded in three-dimensional space,
  at a wavevector~$\q\in\mathbb{C}^3$ is
  \begin{equation}\label{E2ffdef}
    f(\q,\Gamma)
    \coloneqq \iint_\Gamma\!\d^2r_\pa\, \e^{i\q\r},
  \end{equation}
  where $\r\coloneqq \r_\perp+\r_\pa$, $\r_\perp\coloneqq r_\perp\n$,
  and $\r_\pa$ is given by the integration variables.
\end{definition}

\begin{definition}[vector decomposition induced by~$\n$]\label{Ddecompose}
  In a given plane, characterized by a normal vector~$\n$,
  the subscripts $\perp$ and~$\pa$, affixed to a vector~$\x\in\mathbb{C}^3$,
  indicate the projection of~$\x$ onto~$\n$ (\emph{perpendicular} to the plane)
  \begin{equation}
    \x_\perp\coloneqq (\x\n)\n,
  \end{equation}
  and the in-plane (\emph{parallel}) complement
  \begin{equation}\label{Expa}
    \x_\pa\coloneqq\x-\x_\perp,
  \end{equation}
  respectively. The subscript~$\times$ shall denote
  an in-plane vector normal to~$\x$,
  \begin{equation}\label{Excross}
    \x_\times\coloneqq\n\times\x_\pa.
  \end{equation}
\end{definition}

This definition is compatible with, and extends,
the notation of \cref{Dplane,Dff23}
where $\r$ was constructed as $\r_\perp+\r_\pa$.

If it happens that a given vector~$\x$ is exactly or almost parallel to~$\n$,
then the computation of~$\x_\pa$ according to~\cref{Expa} becomes inaccurate,
and can result in $\x_\pa$ having a relatively strong out-of-plane component.
To reduce this spurious component,
one can iterate
  \begin{equation}\label{Expait}
    \x_\pa^{(i+1)}\coloneqq\x_\pa^{(i)}-(\x_\pa^{(i)}\n)\n.
  \end{equation}
Furthermore, if $|\x_\pa|/|\x_\perp|$ is smaller than the machine epsilon,
then it is adequate to let $\x_\pa=\v{0}$.

\begin{proposition}[factorization of $f$]\label{Ptranslate}
  The form factor~\cref{E2ffdef} of a planar figure $\Gamma\subset\Plane(\n,r_\perp)$
  can be factorized as
  \begin{equation}\label{E2fac}
    f(\q,\Gamma)
    = \e^{i\qrperp} f(\qpa,\Gamma).
  \end{equation}
\end{proposition}

\begin{proof}
  Go back to the definition~\cref{E2ffdef} of $f(\q,\Gamma)$:
  \begin{equation}
    f(\q,\Gamma)
    = \iint_\Gamma\!\d^2r_\pa\, \e^{i\q\r}
    = \e^{i\qrperp}  \iint_\Gamma\!\d^2r_\pa\, \e^{i\qpa\r}
    = \e^{i\qrperp} f(\qpa,\Gamma).
  \end{equation}
\end{proof}

The factor~$f(\qpa,\Gamma)$ shall be called the \emph{in-plane form factor} of~$\Gamma$.

\begin{proposition}[continuity at $q_\pa=0$]\label{P2lim0}
  The in-plane form factor of a planar figure $\Gamma\subset\Plane(\n,r_\perp)$
  is continuous at $q_\pa=0$, and has the limit
  \begin{equation}\label{E2lim0}
    \lim_{q_\pa\to0} f(\qpa,\Gamma) = f(\v{0},\Gamma) = \iint_\Gamma\!\d^2r_\pa = \Ar(\Gamma),
  \end{equation}
  where $\Ar(\Gamma)$ is the area of~$\Gamma$.
\end{proposition}

\begin{proof}
  The continuity and the first equality in~\cref{E2lim0} hold
  because for $\q_\pa\to\v{0}$ and $\r\in\Gamma$
  the convergence $\e^{i\q_\pa\r}\to1$ of the integrand of~\cref{Dff23} is uniform.
  The second equality is \cref{Dff23} with $\q_\pa=\v{0}$.
\end{proof}

\begin{proposition}[series expansions]\label{P2ser}
  The form factor of a planar figure $\Gamma$
  possesses for any $\q\in\mathbb{C}^3$
  two different absolutely convergent series expansions
  \begin{subequations} \label{E2sum}
    \begin{align}
    f(\q,\Gamma) \label{E2sum_f}
    &= \sum_{n=0}^\infty i^n f_n(\q,\Gamma),\\
    f(\q,\Gamma) \label{E2sum_p}
    &= \sum_{n=0}^\infty i^n \phi_n(\q,\Gamma)
    \end{align}
  \end{subequations}
  with coefficients
  \begin{subequations}
    \begin{align}
    \label{E2intcoeff}
     f_n(\q,\Gamma)
     &\coloneqq\displaystyle \iint_\Gamma\d^2r_\pa \frac{{(\q\r)}^n}{n!},
    \\
     \phi_n(\q,\Gamma)
     &\coloneqq\displaystyle \e^{i\qrperp}\iint_\Gamma\d^2r_\pa \frac{{(\qpa\r)}^n}{n!},
    \end{align}
  \end{subequations}
  which coincide for in-plane wave vectors, $f_n(\qpa)=\phi_n(\qpa)$.
  For generic~$\q$,
  \begin{subequations}
    \begin{align}
     f_n(\q,\Gamma) \label{E2_fn}
     &=\displaystyle
       \sum_{m=0}^n \frac{{(\qrperp)}^{n-m}}{(n-m)!} f_m(\qpa,\Gamma),\\
     \phi_n(\q,\Gamma) \label{E2_pn}
     &=\displaystyle \e^{i\qrperp} \phi_n(\qpa,\Gamma).
    \end{align}
  \end{subequations}
\end{proposition}


\begin{proof}
  Expand the exponential in the integrand of~\cref{E2ffdef}
  in $\q$ or~$\qpa$ to obtain \cref{E2sum_f} or~\cref{E2sum_p}.
  Define the radii
  \begin{equation}\label{Eadef}
    a\coloneqq\max_j|\V_{j}|
    \quad\text{and\/}\quad
    b\coloneqq\max_j|\V_{j\pa}|
  \end{equation}
  of a sphere respectively a circle that contain all vertices of~$\Gamma$.
  Then
  \begin{subequations}
   \begin{align}
    |f_n(\q,\Gamma)|
    &\displaystyle\le \iint_\Gamma\d^2r_\pa \frac{(qa)^n}{n!}
    = \frac{(qa)^n}{n!} \Ar(\Gamma),
    \\
    |f_n(\q,\Gamma)|
    &\displaystyle\le |\e^{i\qrperp}|\iint_\Gamma\d^2r_\pa \frac{(qb)^n}{n!}
    = \e^{-\text{Im}(q_\perp)r_\perp}\frac{(qb)^n}{n!} \Ar(\Gamma)
    \end{align}
  \end{subequations}
  proves the absolute convergence of both series.
  The factorization of $f_n$ in~\cref{E2_fn} is obtained
  by straightforward binomial expansion of $(\qpa\r+\q_\perp\r)^n$.
\end{proof}

\begin{remark}[choice of origin]
  It may be advisable to translate the origin before starting a numeric computation of~$f$.
  This can be done using a standard property of Fourier transforms
\begin{equation}\label{EFTtra}
  f(\q,\Gamma) = \e^{i\q\v{v}} f(\q,\Gamma-\v{v}).
\end{equation}
Two particular choices of the origin have special advantages:
An origin at the center of the enclosing circle minimizes~$a$,
whereas an origin at the center of gravity lets the expansion coefficient~$f_1$ vanish.
\end{remark}

\begin{remark}[termination of numeric summation]\label{R2terminate}
  In general, the integral formula \cref{E2intcoeff} is no practicable starting point
  for computing the~$f_m(\qpa,\Gamma)$.
  It allows us, however, to make the following observation:
  For odd~$m$, but not for even~$m$, it may happen that $f_m=0$
  for some~$\qpa\ne\v{0}$.
  In particular, $f_1\equiv0$ for all $\qpa$ if the origin is chosen at the center of gravity.
  Therefore the termination criterion in a numeric implementation of
  \cref{E2sum_f} or~\cref{E2sum_p}
  must only rely on terms with even~$m$.
\end{remark}

\section{Form factor of a polygon}\label{S2poly}

\begin{definition}[simple polygonal vertex chain]\label{DV}
  A \emph{simple polygonal vertex chain} of size~$J$
  is a sequence of points $\V_1,\ldots,\V_J$,
  made cyclic by the convention $\V_0\equiv\V_J$,
  that lie all in a plane,
  and when connected by edges $\overline{\V_0\V_1},\ldots,\overline{\V_{J-1},\V_J}$,
  concatenated in this order, form a non-intersecting loop.

  This loop is the border~$\partial\Gamma$ of a polygon $\Gamma(\V_1,\ldots,\V_J)$.

  The plane $\Plane(\n,r_\perp)$, given by $\V_1,\ldots,\V_J$,
  is oriented such that $\partial\Gamma$ has a winding number~$+1$
  with respect to~$\n$.
\end{definition}

Throughout the following, we assume that the origin is chosen such that
$\V_j\ne\v{0}$ for all~$j\in\{1,\ldots,J\}$.
For a given polygonal vertex chain,
the normal vector can then be computed as
\begin{equation}\label{En}
  \n = \frac{\V_{j-1}\times\V_{j}}{|\V_{j-1}\times\V_{j}|},
\end{equation}
and the location of the plane as
\begin{equation}\label{Erp}
  r_\perp = \n \V_j,
\end{equation}
with arbitrary $j\in\{1,\ldots,J\}$.
If vertex coordinates are ill-conditioned,
then it may be indicated to use \cref{En,Erp} with a special choice of~$j$,
or to average over several~$j$.

\begin{figure}[t]\label{Ftriangle}
  \centering
  \includegraphics[width=.47\linewidth]{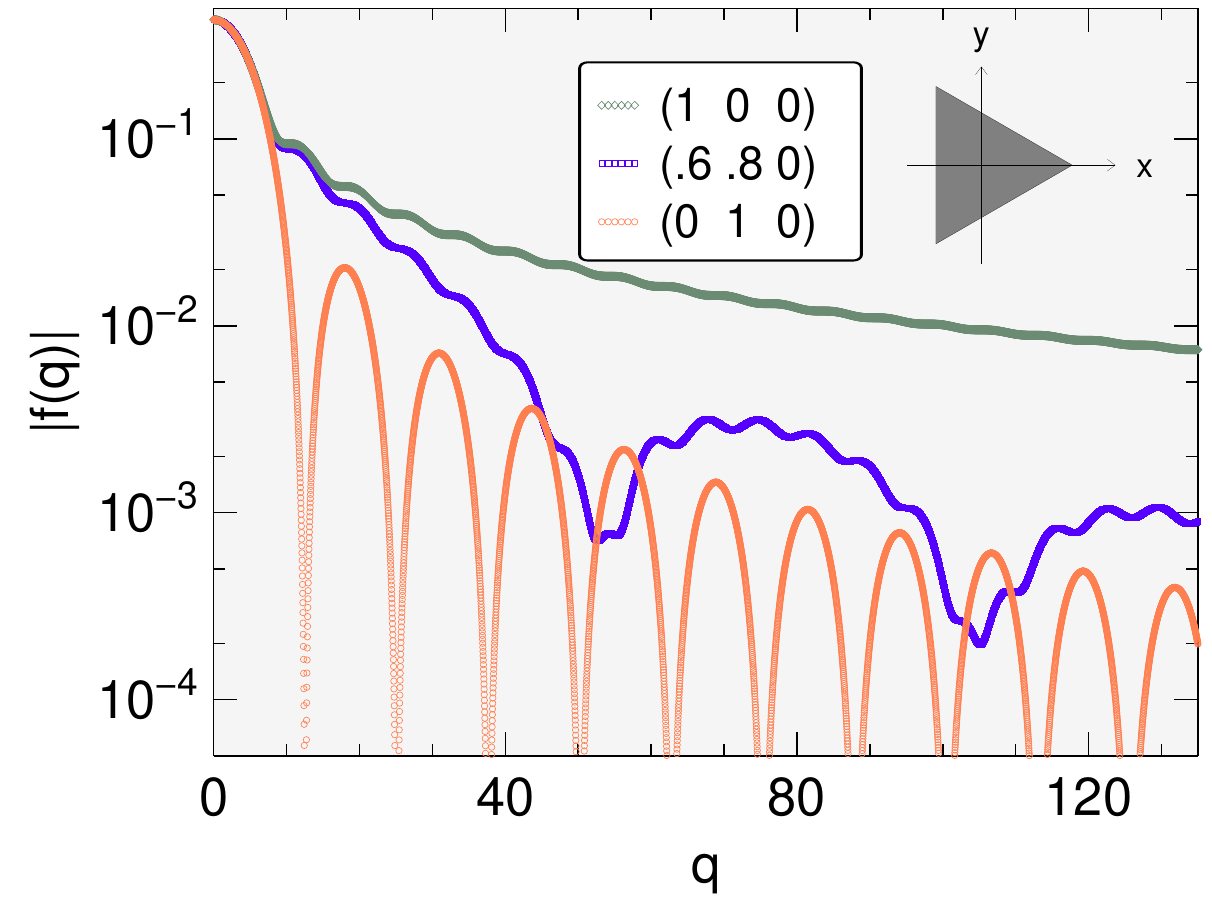}
  \hfill
  \includegraphics[width=.47\linewidth]{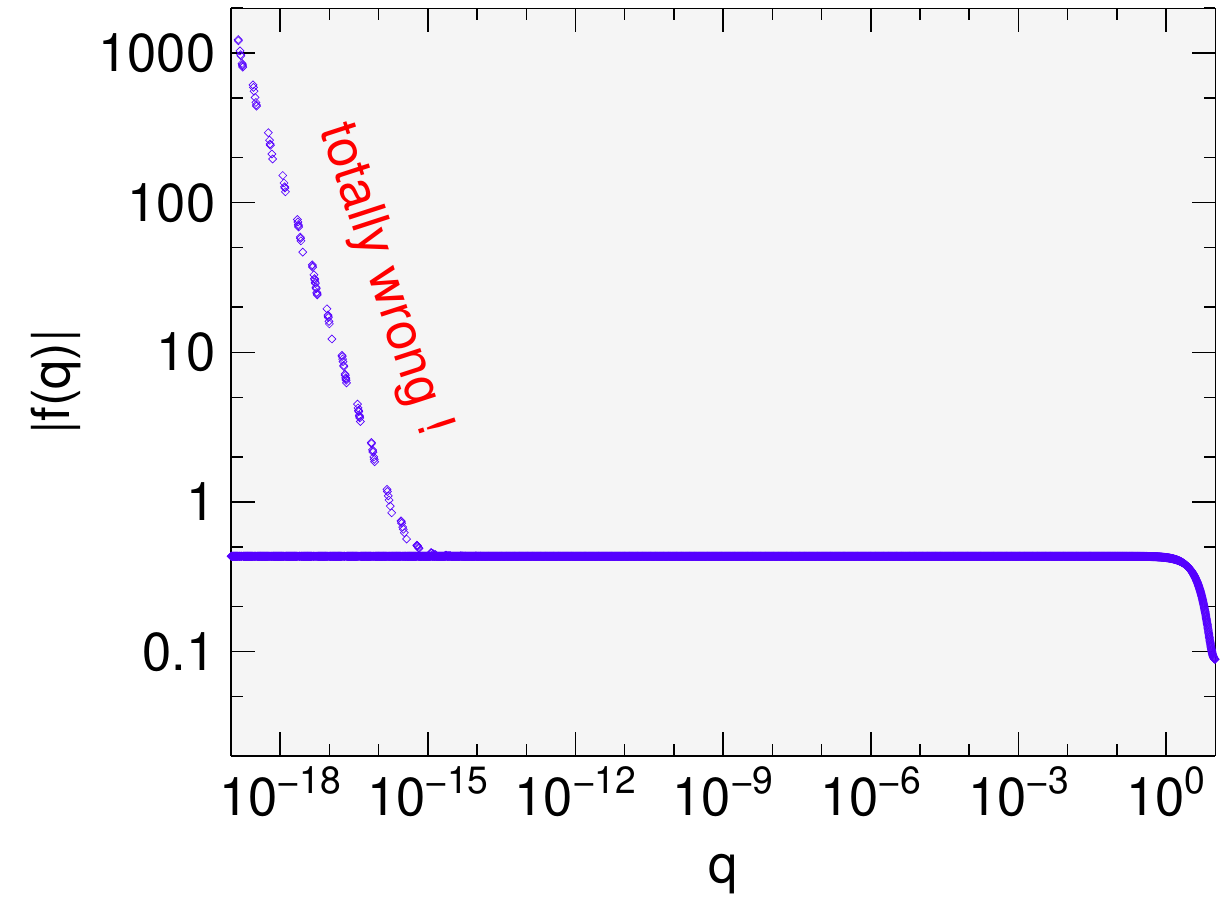}
  \caption{Form factor of an equilateral triangle
    as function of wavenumber~$q$ for three different wavevector directions $\v{\hat q}$,
    computed using the analytic expression \cref{E2ne0} of \cref{P2d}.
    Geometry: The triangle lies in the $xy$ plane.
    It is oriented so that an edge points in $x$ direction.
    The edge length is $L=1$.
    The second panel shows gross numeric errors at very small~$q$ for $\v{\hat q}=(.6,.8,0)$,
    as explained in \cref{L2const}.
}
\end{figure}

\begin{proposition}[form factor of a polygon]\label{P2d}
  A $J$-gon~$\Gamma\subset\Plane(\n,r_\perp)$
  be given by a simple polygonal vertex chain~$\V_1,\ldots,\V_{J}$.
  We abbreviate $\E_j\coloneqq(\V_{j}-\V_{j-1})/2$ and $\R_j\coloneqq(\V_{j}+\V_{j-1})/2$.
  Then the form factor for~$q_\pa\ne0$ is
  \begin{equation}\label{E2ne0}
    f(\q,\Gamma)
    = \frac{2}{iq_\pa^2}\qcross^*\cdot
      \sum_{j=1}^{J} \E_j \sinc(\q\E_j) \e^{i\q\R_j}.
  \end{equation}
\end{proposition}

\begin{proof}
  We parametrize the edges of the polygon by $\r_j(\lambda)\coloneqq \R_j + \E_j\lambda$
  with $-1\le\lambda\le+1$.
  Stokes's theorem then takes the form
  \begin{equation}\label{EStokes}
    \iint_\Gamma\!\d r^2\,\n\cdot(\Nabla\times\v{G})
     = \oint_{\partial\Gamma} \d\v{r}\cdot\v{G}
     = \sum_{j=1}^{J} \int_{-1}^{+1}\!\d\lambda\, \frac{\d\r_j}{\d\lambda}\,\v{G}
     = \sum_{j=1}^{J} \E_j \int_{-1}^{+1}\!\d\lambda\,\v{G}.
  \end{equation}
  With the choice $\v{G}\coloneqq\qcross^*\, \e^{i\q\r}$, this yields
  \begin{equation}
    iq_\pa^2 f(\q,\Gamma)
    = \qcross^*\cdot
    \sum_{j=1}^{J} \E_j \int_{-1}^{+1}\!\d\lambda\, \e^{i\q\r}.
  \end{equation}
  With $q_\pa\ne0$:
  \begin{equation}\label{E2alt1a}
      f(\q,\Gamma)
      = \displaystyle\frac{1}{iq_\pa^2}\qcross^*\cdot
      \sum_{j=1}^{J} \E_j \frac{\e^{i\q(\R_j+\E_j)}-\e^{i\q(\R_j-\E_j)}}{i\q\E_j},
  \end{equation}
  which can easily be brought into the form~\cref{E2ne0}.
\end{proof}

\Cref{Ftriangle} illustrates the so determined form for about the simplest polygon,
an equilateral triangle.
$|f(q)|$ is plotted as function of $q$ for three different directions~$\v{\hat q}$.

\begin{remark}[sinc is accurate]\label{Rsinc}
  The cardinal sine function
  \begin{equation}\label{Esinc}
    \sinc(z) \coloneqq
    \left\{\begin{array}{ll}
    \displaystyle \frac{\sin(z)}{z}
    &\text{~for~}z\ne0,\\[2.3ex]
    \displaystyle 1
    &\text{ otherwise,}
    \end{array}\right.
  \end{equation}
  is best implemented by literally following this definition.
  Any available implementation of $\sin(z)$ will have full floating-point accuracy
  for $|z|\to0$, and so $\sin(z)/z$ will have.
\end{remark}

\begin{remark}[edge-independent terms cancel]\label{L2const}
  Since the vertices of~$\Gamma$ form a closed loop,
  the sum of the edge vectors vanishes,
  \begin{equation}\label{E2const}
    \sum_{j=1}^{J} \E_j = 0.
  \end{equation}
  The sum over~$j$ in~\cref{E2ne0} contains a factor $\sinc(\q\E_j)\exp(i\q\R_j)$.
  The leading term in a $q_\pa$-expansion of this factor is $\exp(i\qrperp)$.
  It can be drawn in front of the sum,
  which then vanishes per~\cref{E2const}.
  If the next-to-leading terms are relatively small,
  then the cancellation of the leading term will cause a loss of accuracy in the resulting
  form factor.
  Therefore, the analytic result~\cref{E2ne0} is not practicable for small~$q_\pa$,
  where it must be replaced by the series expansion~\cref{E2sum_p}.
  Alternatively, if not only $q_\pa$ but also $q_\perp$ is small,
  then the $q$-expansion~\cref{E2sum_f} can be used.
  Expansion coefficients will be provided below in~\cref{P2coeff}.
\end{remark}

The second panel of \cref{Ftriangle} demonstrates how this cancellation shows up
in a double-precision implementation of~\cref{E2ne0}.
For $q_\pa L$ close to or below the machine epsilon of $2\cdot10^{-16}$,
resulting values for some~$q_\pa$ are wrong by $\mathcal{O}((Lq_\pa)^{-1})$.

\begin{remark}[area formula]
  The case $q_\pa=0$ is covered by \cref{P2lim0}.
  The area of a polygon can be conveniently computed
  using the \textit{surveyor's formula} \cite{Bra86},
  \begin{equation}\label{ESurvey}
    \Ar(\Gamma) = \frac{1}{2}\,\n\cdot\sum_{j=1}^{J} \V_{j-1}\times\V_{j},
  \end{equation}
  which is based on a triangular tesselation.
\end{remark}

\begin{remark}[relation to literature result]\label{Rlit}
  A closed expression for the form factor of the polygon
  is known since long~\cite[Eq.~6]{LeMi83}.
  In our notation and after a few obvious simplifications, it reads
  \begin{equation}\label{ELeMi}
    f(\q,\Gamma) = \n\cdot\sum_{j=1}^{J}
    \frac{\E_{j-1}\times\E_{j}}{(\q\E_{j-1})(\q\E_{j})} \e^{i\q\V_j}.
  \end{equation}
  This expression is beautiful for being short and symmetric.
  However, for each~$j$   there are two $\q$~planes for which the denominator vanishes.
  A practical implementation of~\cref{ELeMi} that takes proper care of these singularities
  would be more complicated than an algorithm based on our \cref{P2d}.
\end{remark}

\begin{proof}
  Let us demonstrate the equivalence of~\cref{ELeMi} with our Proposition~\ref{P2d}.
  Start from \Cref{E2alt1a}.
  In the exponential functions, insert the definitions of $\E_j$ and $\R_j$:
  \begin{equation}
    f(\qpa,\Gamma)
    =  \frac{\qcross^*}{iq_\pa^2}\cdot
    \sum_{j=1}^{J} \frac{\E_j}{i\qpa\E_j}
    \left( \e^{i\qpa\V_{j}}-\e^{i\qpa\V_{j-1}}\right).
  \end{equation}
  Shuffle indices~$j+1\to j$, use $\q\E_j=\qpa\E_j$,
  and employ standard vector identities:
  \begin{subequations}
    \begin{align}
      f(\q,\Gamma)
      &= \displaystyle  \frac{\qcross^*}{iq_\pa^2}\cdot
      \sum_{j=1}^{J}
      \left( \frac{\E_{j-1}}{i\q\E_{j-1}} - \frac{\E_j}{i\q\E_j} \right) \e^{i\q\V_j}\\
\label{Evecid_i}
      &= \displaystyle - \frac{\qcross^*}{q_\pa^2}\cdot
      \sum_{j=1}^{J}
      \frac{\E_{j-1}(\qpa\E_{j})-\E_{j}(\qpa\E_{j-1})}
           {(\q\E_j)(\q\E_{j-1})} \e^{i\q\V_j}\\
      &= \displaystyle - \frac{\qcross^*}{q_\pa^2}\cdot
      \sum_{j=1}^{J}
      \frac{\qpa\times(\E_{j}\times\E_{j-1})}
           {(\q\E_j)(\q\E_{j-1})} \e^{i\q\V_j}\\
\label{Evecid_f}
      &= \displaystyle - \frac{\qcross^*\times\qpa}{q_\pa^2}\cdot
      \sum_{j=1}^{J}
      \frac{\E_{j}\times\E_{j-1}}
           {(\q\E_j)(\q\E_{j-1})} \e^{i\q\V_j}
    \end{align}
  \end{subequations}
  Insert the definition of~$\qcross$, employ another vector identity,
  use $\n\qpa=0$, and obtain~\cref{ELeMi}.
\end{proof}

\begin{figure}[t]\label{F2ser}
  \centering
  \includegraphics[width=.47\linewidth]{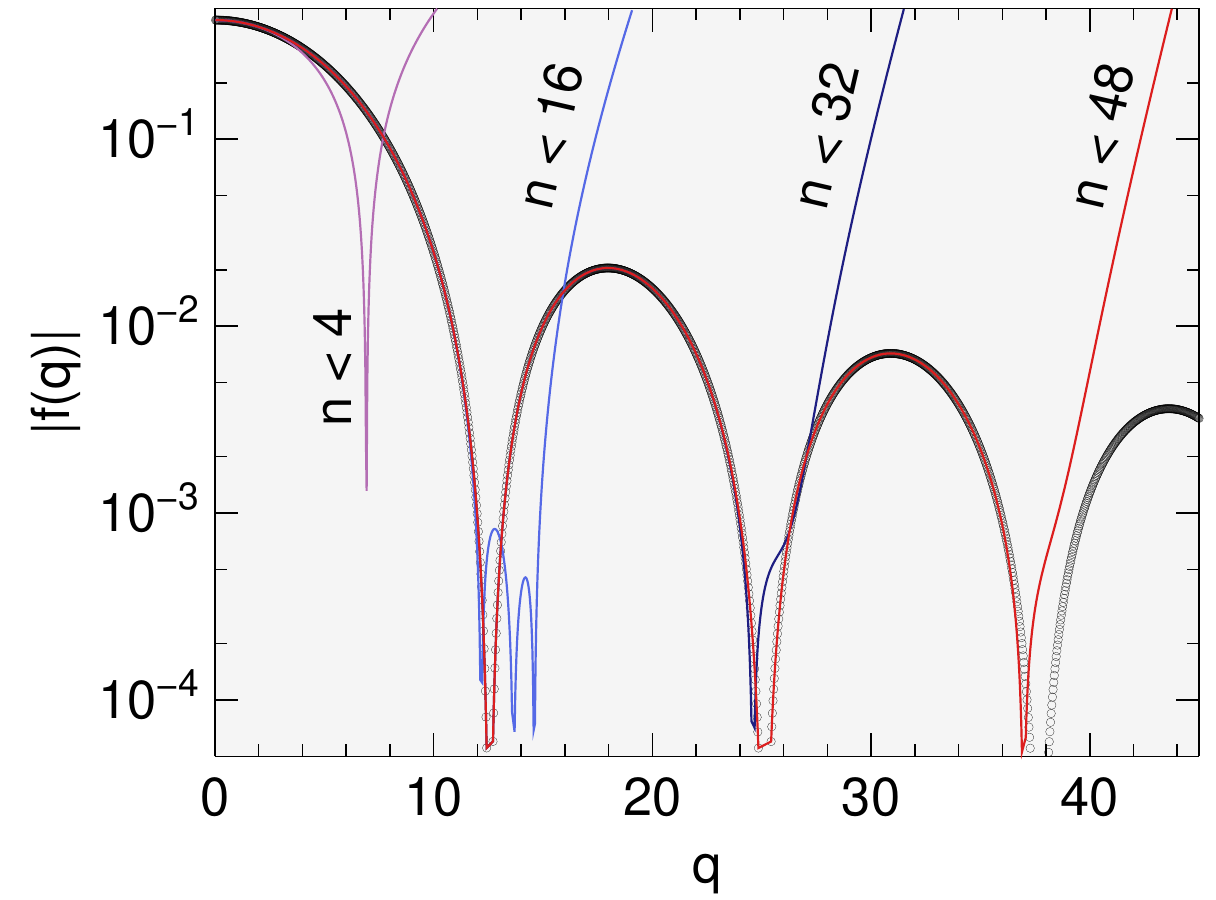}
  \caption{Form factor of the equilateral triangle of \cref{Ftriangle},
    as function of wavenumber~$q$ for wavevector direction $\v{\hat q}=(0,1,0)$.
    The chain of black circles is computed using the analytic expression \cref{E2ne0}.
    The colored curves are computed using the series expansion~\cref{E2sum_f}
    with coefficients~\cref{E2coeff} up to the indicated orders~$n$.
}
\end{figure}

\begin{proposition}[expansion coefficients]\label{P2coeff}
  The coefficients of the series expansion~\cref{E2sum_f}
  of the form factor can be computed as
  \begin{equation}\label{E2coeff}
    f_n(\q,\Gamma)
    = \frac{2}{q_\pa^2}\, \qcross^*\cdot \sum_{j=1}^{J} \,\E_j
    \sum_{l=0}^{(n+1)/2}
    \frac{{\left(\q\E_j\right)}^{2l}}{(2l\!+\!1)!}
    \frac{{\left(\q\R_j\right)}^{n+1-2l}}{(n\!+\!1\!-\!2l)!}.
  \end{equation}
\end{proposition}

\begin{proof}
  Expand the functions $\sinc(\q\E)$ and $\exp(i\q\R)$ in~\cref{E2ne0}:
  \begin{equation}
    f(\q,\Gamma)
    = \frac{2}{i q_\pa^2} \qcross^*\cdot \sum_{j=1}^{J} \,\E_j\,
    \sum_{l=0}^\infty i^{2l} \frac{{\left(\q\E_j\right)}^{2l}}{(2l\!+\!1)!}
    \sum_{\mu=0}^\infty \frac{\left(i\q\R_j\right)^\mu}{\mu!}.
  \end{equation}
  Sort by powers of $q$:
  \begin{equation}\label{E2ser8}
    f(\q,\Gamma)
    = \sum_{\nu=0}^\infty i^{\nu-1} \:
     \frac{2}{q_\pa^2}\, \qcross^*\cdot \sum_{j=1}^{J} \,\E_j\,
    \sum_{l=0}^{2l\le \nu} \frac{{\left(\q\E_j\right)}^{2l}}{(2l\!+\!1)!}
    \frac{{\left(\q\R_j\right)}^{\nu-2l}}{(\nu\!-\!2l)!}.
  \end{equation}
  For $\nu=0$, the sum over~$l$ yields the constant~1,
  which cancels per \cref{E2const} under the sum over~$j$.
  Therefore the outer sum may as well start at~$\nu=1$.
  Substitute $n=\nu-1$ to obtain the series~\cref{E2sum_f} with coefficients~\cref{E2coeff}.
\end{proof}

\Cref{F2ser} shows that the series expansion works well even beyond the first few oscillations
in $f(|q|)$, provided a sufficient number of terms is summed.

\begin{remark}[area formula from expansion coefficient]\label{R2dir}
  From \cref{E2intcoeff} we know that $f_0(\q,\Gamma)=\Ar(\Gamma)$.
  On the other hand, per \cref{P2coeff},
  \begin{equation}\label{E2coef0}
    f_0(\q,\Gamma)
    = \frac{2}{q_\pa^2} \qcross^*\cdot \sum_{j=1}^{J} \E_j\: (\q\R_j).
  \end{equation}
  It was not necessary for the above proofs,
  but may nevertheless be interesting to explicitly show how \cref{E2coef0}
  can be algebraically transformed into the $\q$-independent form~\cref{ESurvey}.

  Insert the definitions of $\E_j$ and $\R_j$ into~\cref{E2coef0}:
  \begin{equation}
    f_0(\q,\Gamma)
    = \displaystyle \frac{\qcross^*}{2q_\pa^2}\cdot\sum_{j=1}^{J}
    (\V_{j}-\V_{j-1}) (\q(\V_{j}+\V_{j-1})).
  \end{equation}
  Multiply out, and shuffle indices $(j\!-\!1)\to j$ for some terms under the sum:
  \begin{equation}\label{Ecrossdotdiff}
    f_0(\q\Gamma)
    = \displaystyle \frac{\qcross^*}{2q_\pa}\cdot\sum_{j=1}^{J}
    \left[\V_{j} (\q\V_{j-1}) -\V_{j-1} (\q\V_{j})\right].
  \end{equation}
  From here, simplify using the same standard vector identities as were employed in \cref{Rlit}.
  Use $f_0(\q,\Gamma)=\Ar(\Gamma)$ to obtain~\cref{ESurvey}.
\end{remark}

\begin{remark}[resulting algorithm]\label{R2final}
  Let us summarize \cref{S2fig,S2poly} by outlining an algorithm for reliably computing
  the form factor of any polygon~$\Gamma$, given through an oriented vertex chain.
  The algorithm involves preselected constants $c$ and $c_\pa$
  that determine when to use a series expansion instead of the analytic formula.
  Do not use this algorithm if the symmetry of~$\Gamma$ allows for a simpler computation,
  as discussed below in \cref{Si}.

  If the input coordinates cannot be trusted:
  Test whether all vertices lie in a plane;
  otherwise terminate with error message.
  Check whether the coordinate origin lies well inside~$\Gamma$;
  otherwise determine a new origin, for instance from the center of gravity of~$\Gamma$,
  and apply \cref{EFTtra} at the end of the computation.

  Compute $\E_j$, $\R_j$, $\n$ and $r_\perp$, and decompose~$\q$ into $\qpa+\q_\perp$.
  If $q_\pa=0$, then return $f(\q,\Gamma)=\e^{i\qrperp}\Ar(\Gamma)$.
  If $a q<c$, then sum the series~\cref{E2sum_f} with $f_0(\q,\Gamma)=\e^{i\qrperp}\Ar(\Gamma)$,
  and other coefficients $f_n(\q,\Gamma)$ given by~\cref{E2coeff}.
  Terminate the summation and return the result if
  a term for even~$n$, relative the absolute value of the sum acquired so far,
  falls below the machine epsilon.
  If $a q_\pa<c_\pa$, then sum the series~\cref{E2sum_p} with coefficients $f_n(\qpa,\Gamma)$
  still given by~\cref{E2coeff}.
  Apply the same termination criterion as in the previous case, and return the result.
  In the remaining cases, return the result of the direct computation~\cref{E2ne0}.
\end{remark}

\section{Form factor of a three-dimensional figure}\label{S3fig}

\begin{definition}[form factor]
  The \emph{form factor} of a figure~$\Pi\subset\mathbb{R}^3$
  at a \emph{wavevector}~$\q\in\mathbb{C}^3$ is the Fourier transform
  of its indicator function, given by the integral
  \begin{equation}\label{E3ffdef}
    F(\q,\Pi)
    = \iiint_\Pi\!\d^3r\, \e^{i\q\r}.
  \end{equation}
\end{definition}

\begin{proposition}[continuity at $q=0$]\label{P3lim0}
  The form factor of a figure~$\Pi$ is continuous at $q=0$, and has the limit
  \begin{equation}\label{E3lim0}
    \lim_{q\to0} F(\qpa,\Pi) = F(\v{0},\Pi) = \iiint_\Pi\!\d^3r = \Vol(\Pi),
  \end{equation}
  where $\Vol(\Pi)$ is the volume of~$\Pi$.
\end{proposition}

The proof is fully analogous to that of \cref{P2lim0}.
The later Lord Rayleigh had already noted about
the scattering of light by small particles \cite{Ray71b}:
``The leading term, we have seen, depends only on the
volume ; but the same would not be true for those that follow \ldots
there is no difficulty in proceeding further; but I have not arrived at any results of interest.''

\begin{proposition}[series expansion]\label{P3ser}
  The form factor of a figure $\Pi$
  possesses for any $\q\in\mathbb{C}^3$ the absolutely convergent series expansion
  \begin{equation}\label{E3series}
    F(\q,\Pi)
    = \sum_{n=0}^\infty i^n F_n(\q,\Pi)
  \end{equation}
  with coefficients
  \begin{equation}\label{E3intcoeff}
     F_n(\q,\Pi)
     \coloneqq \iiint_\Pi\d^3r \frac{{(\q\r)}^n}{n!}.
  \end{equation}
\end{proposition}

\begin{proof}
  Expand the exponential function in the integrand of~\cref{E3ffdef}
  to obtain \cref{E3series}.
  Use the radius~$a$ defined in \cref{Eadef}
  to derive the bound
  \begin{equation}
    |F_n(\q,\Pi)|
    \le \iiint_\Pi\d^3r \frac{(qa)^n}{n!}
    = \frac{(qa)^n}{n!} \Vol(\Pi).
  \end{equation}
  This proves the absolute convergence of the series
  and thereby the existence and uniqueness of~\cref{E3series}.
\end{proof}

\begin{remark}[termination of numeric summation]\label{R3terminate}
  The same considerations as in \cref{R2terminate} apply:
  coefficients~$F_n$ may vanish for odd~$n$;
  therefore the termination criterion in a numeric implementation of~\cref{E3series}
  must only rely on terms with even~$n$.
\end{remark}

\section{Form factor of a polyhedron}\label{S3poly}

\begin{proposition}[form factor of a polyhedron]\label{P3d}
  An orientable polyhedron~$\Pi$ be given by its $K$ faces~$\Gamma_k$
  ($k=1,\ldots,K$).
  Each face $\Gamma_k\in\Plane(\n_k,r_{\perp k})$ be an $J_k$-gon,
  given by the simple polygonal vertex chain~$\V_{k1},\ldots,\V_{kJ_k}$,
  and with $\n_k$ pointing towards the outside of\/~$\Pi$.
  Then the form factor for $q\ne0$ is
  \begin{equation}\label{E3ne0}
    F(\q,\Pi)
    = \frac{1}{iq^2}\, \q^*\cdot \sum_{k=1}^K \n_k \:f(\q,\Gamma_k).
  \end{equation}
\end{proposition}

\begin{proof}
  We first address the case $q\ne0$.
  For a polyhedron, the divergence theorem takes the form
  \begin{equation}\label{EGauss}
    \iiint_\Pi\!\d^3r\,\Nabla\v{H}
     = \iint_{\partial\Pi}\!\d^2r\, \n\v{H}
     = \sum_k \n_k \iint_{\Gamma_k}\!\d^2r_\pa\, \v{H}.
  \end{equation}
  With the choice $\v{H}\coloneqq\q^*\,\e^{i\q\r}$, this yields
  \begin{equation}
    i q^2 F(\q,\Pi)
    = \q^*\cdot \sum_{k=1}^{K} \n_k \iint_{\Gamma_k}\!\d^2r_\parallel\, \e^{i\q\r}.
  \end{equation}
  With the notation \cref{E2ffdef}, this proves~\cref{E3ne0}.
\end{proof}

\begin{remark}[volume formula]
  The case $q=0$ is covered by \cref{P2lim0}.
  The volume of a polyhedron can be conveniently computed from
  \begin{equation}\label{E3Vol}
    \Vol(\Pi)
    = \sum_k \n_k \iint_{\Gamma_k}\d^2r_\parallel\,\r/3
    = \frac{1}{3}\sum_k \Ar(\Gamma_k)\,r_{\perp k}.
  \end{equation}
  This formula, which is based on a tetrahedral tesselation.
  can be retrieved from remote literature~\cite[Cap.~II, \S~3, III 171]{Com30},
  or easily be derived by applying the divergence theorem~\cref{EGauss}
  with the choice $\v{H}\coloneqq\r/3$.
\end{remark}

\begin{remark}[face-independent contributions cancel]\label{L3const}
  For a closed surface, the integral of surface normals cancels.
  For a polyhedron, the integral becomes a sum:
  \begin{equation}\label{E3const}
    \sum_{k=1}^{K} \n_k \Ar(\Gamma_k) = 0.
  \end{equation}
  For proof, use the divergence theorem~\cref{EGauss} with a $k$- and $\r$-independent~$\v{H}$.

  In full analogy with~\cref{L2const},
  this cancellation can lead to a loss of accuracy in~\cref{E3ne0}
  if $f(\q,\Gamma_k)/\Ar(\Gamma_k)$ is dominated by
  $k$-independent terms.
  According to the series expansion~\cref{E2sum_f},
  this is indeed the case; the leading term is just
  $f_0(\q,\Gamma_k)=\Ar(\Gamma_k)$.
  Therefore, the analytic result~\cref{E3ne0} is not practicable for small~$q$,
  where it must be replaced by the series expansion~\cref{E3series}.
  Coefficients, unaffected by cancellation, will be provided below in~\cref{P3coeff}.
\end{remark}

\begin{proposition}[expansion coefficients]\label{P3coeff}
  The coefficients of the series expansion \cref{E3series}
  can be computed as
  \begin{equation}\label{E3coeff}
    F_n(\q,\Pi)
    = \frac{1}{q^2} \q^*\cdot \sum_{k=1}^K \n_k f_{n+1}(\q,\Gamma_k).
  \end{equation}
\end{proposition}

\begin{proof}
  Insert~\cref{E2sum_f} in~\cref{E3ne0}:
  \begin{equation}
    F(\q,\Pi)
    = \frac{1}{iq^2} \q^*\cdot \sum_{k=1}^K \n_k \sum_{\nu=0}^\infty i^\nu f_\nu(\q,\Gamma_k).
  \end{equation}
  Per \cref{P2ser}, the $f_\nu$ are absolutely convergent.
  Therefore the two sums can be exchanged.
  Per \cref{L3const}, the term $f_0(\q,\Gamma_k)=\Ar(\Gamma_k)$ cancels under the sum over~$k$.
  Therefore the infinite sum may as well start at~$\nu=1$.
  Substitute $n=\nu-1$ to obtain the series~\cref{E3series} with coefficients~\cref{E3coeff}.
\end{proof}

\begin{figure}[t]\label{Fmatch}
  \centering
  \includegraphics[width=.47\linewidth]{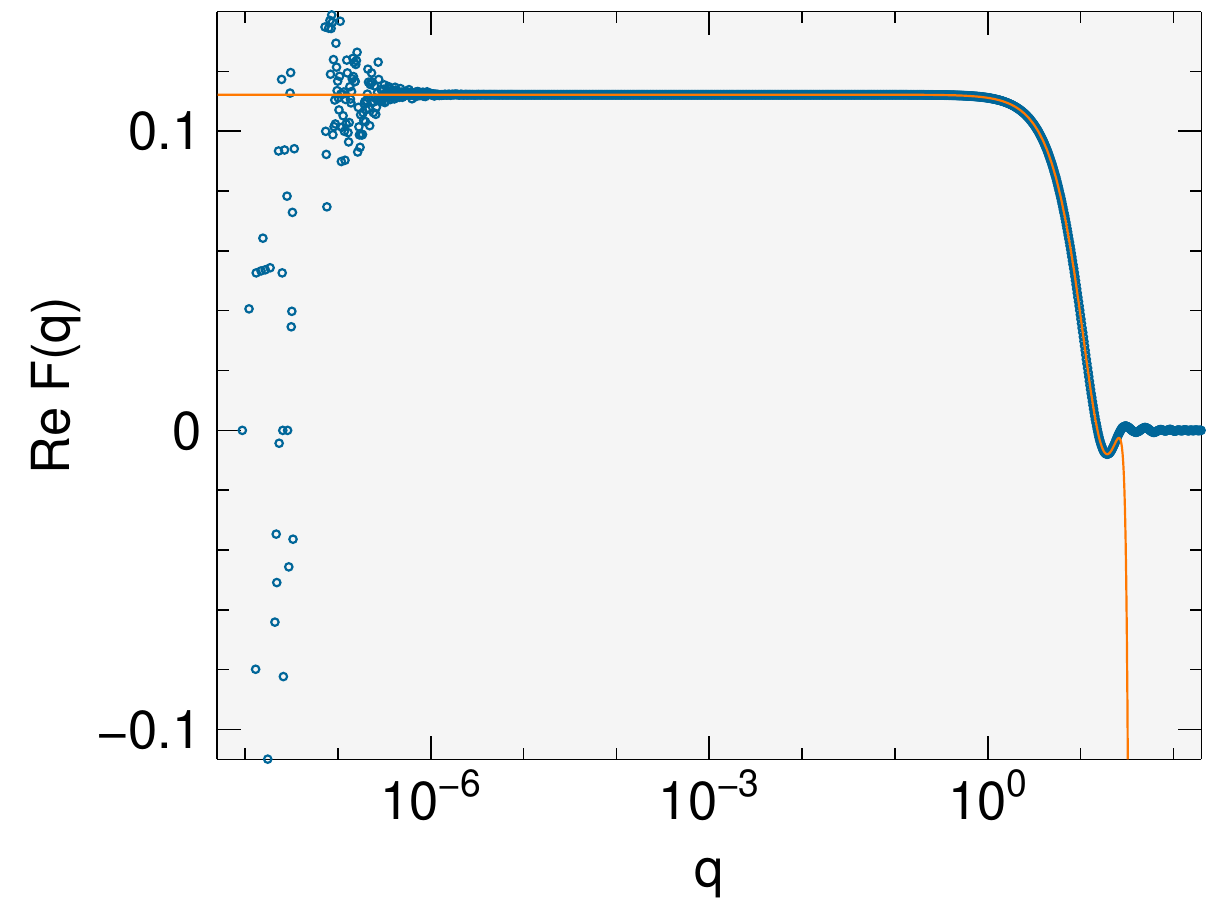}
  \hfill
  \includegraphics[width=.47\linewidth]{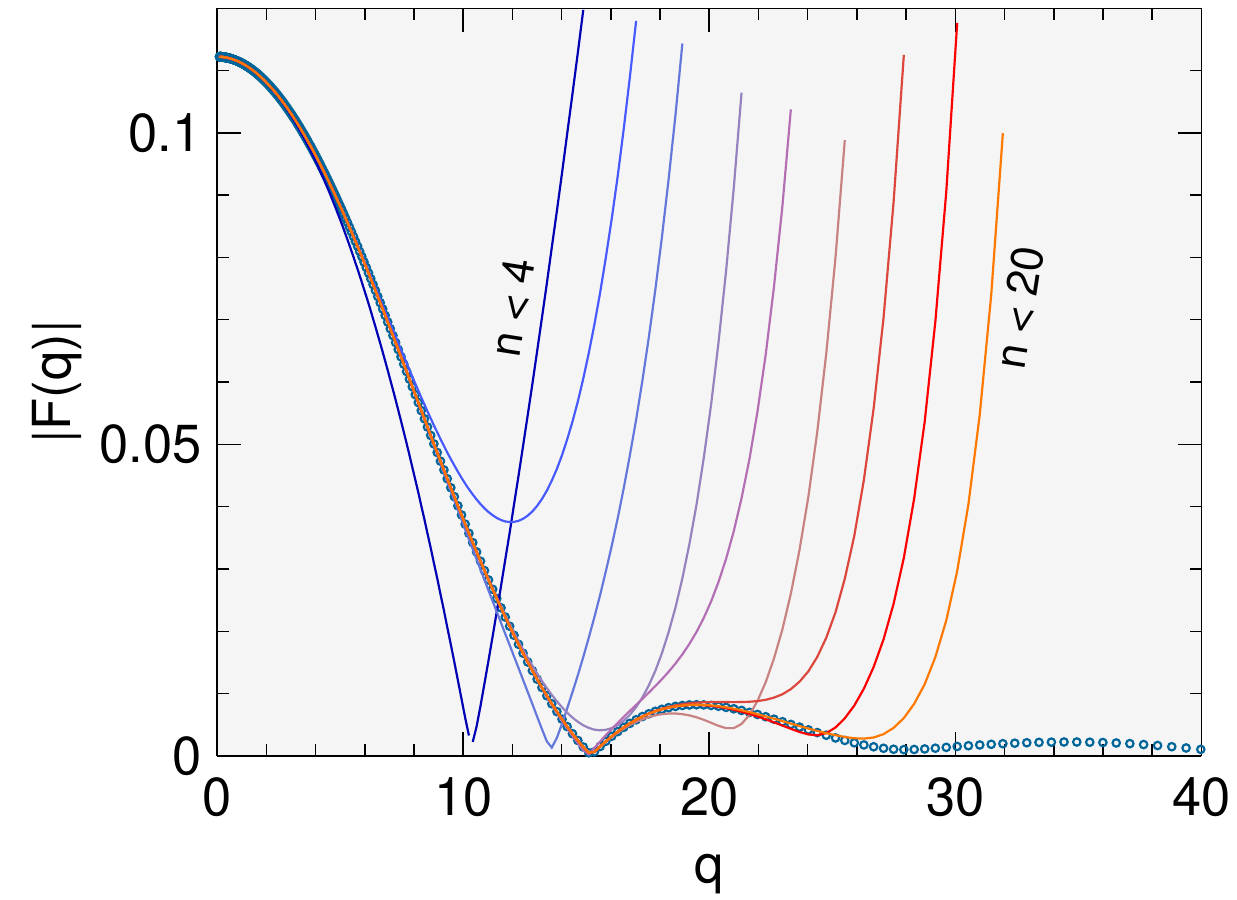}
  \caption{Form factor of a truncated tetrahedron (trigonal pyramidal frustum)
as function of wavenumber~$q$
[The base is an equilateral triangle in the $xy$ plane,
oriented so that an edge points in $y$ direction,
with edge length $L=1$;
the dihedral angle is $72^\circ$; the height $H=L/2$.
Wavevectors are real and off-symmetric, $\v{\hat q}=(1,2,3)/\sqrt{14}$].
(a) $\Real F(\q)$ vs.~$q$. Blue spots computed using the analytic expression
\cref{E2ne0,E3ne0}. The orange line is sum~\cref{E3series} with coefficients~\cref{E3coeff},
ignoring the termination criterion and summing terms up to $n=19$.
(b) $|F(\q)|$ vs.~$q$, now on a linear scale. Blue symbols as before.
Lines are sum~\cref{E3series} up to $n=3,5,\ldots,19$.
}
\end{figure}

\Cref{Fmatch} illustrates the crossover between the domains of analytic computation
and series expansion.
The implementation, described below in \cref{Stests}, uses IEEE double-precision arithmetics.
The analytic form factor~\cref{E3ne0} with per-face contributions~\cref{E2ne0}
fails grossly for $qa\lesssim10^{-5}$,
as was to be expected from \cref{L3const}.
On the other side,
the series expansion~\cref{E3series} requires more and more terms if $qa\gtrsim10$.
This leaves a wide intermediate range where both methods agree well enough
to allow a match with decent accuracy, as further discussed in \cref{Stests}.

\begin{remark}[resulting algorithm]\label{R3final}
  Let us summarize \cref{S3fig,S3poly} by outlining an algorithm for reliably computing
  the form factor of any given polyhedron~$\Pi$,
  given through a set of $K$ polygonal faces~$\Gamma_k$.
  The algorithm involves a preselected constant $C$
  that determines when to use the series expansion instead of the analytic formula.
  Do not use this algorithm if the symmetry of~$\Pi$ allows for a simpler computation,
  as discussed below in \cref{Si}.

  If the input coordinates cannot be trusted,
  check whether the origin lies well inside~$\Pi$.
  Otherwise determine a new origin, for instance from the center of gravity of~$\Pi$,
  and apply \cref{EFTtra} at the end of the computation.

  If $q=0$, then return $F(\v{0},\Pi)=\Vol(\Pi)$.
  If $aq<C$, then sum the series~\cref{E3series},
  using $F_0=\Vol(\Pi)$ and other coefficients from~\cref{E3coeff}.
  Terminate the summation and return the result
  if a term for even~$n$, relative the absolute value of the sum acquired so far,
  falls below the machine epsilon.
  Otherwise compute $F(\q,\Pi)$ according to~\cref{E3ne0} as a weighted sum over
  face form factors. To compute these $f(\q,\Gamma_k)$,
  follow the algorithm described in \cref{R2final}.
\end{remark}

\section{Symmetric figures}\label{Si}

For figures with certain symmetries,
the form factor computation can be considerably simplified.
Besides the speed benefit, this also improves the accuracy
by setting cancelling terms exactly to zero.

\begin{remark}[polygon with symmetry~$S_2$]\label{R2inv}
  If a planar $(2\JJ)$-gon $\Gamma$ has a perpendicular twofold symmetry axis
(Schoenflies group $S_2$), and thereby an inversion center at point~$\r_\perp$,
  then its form factor has the symmetry $f(\q_\perp+\qpa)=f(\q_\perp-\qpa)$.
  Its real-space coordinates transform as
  $\V_{j+\JJ\pa}=-\V_{j\pa}$ for $j=1,\ldots,\JJ$,
  and thereby $\R_{j+\JJ\pa}=-\R_{j\pa}$ and $\E_{j+\JJ}=-\E_{j}$.
  This allows the following simplifications:

  The area can be computed as
  \begin{equation}\label{ESurvey_inv}
    \Ar(\Gamma) = \n\cdot\sum_{j=1}^{\JJ} \V_{j-1}\times\V_{j}.
  \end{equation}
  The form factor for $q_\pa\ne0$ can be computed as
  \begin{equation}\label{E2ne0inv}
    f(\q,\Gamma)
    = \frac{4}{q_\pa^2}\e^{i\qrperp}\qcross^*\cdot
       \sum_{j=1}^{\JJ} \E_j \sinc(\q\E_j) \sin(\qpa\R_j).
  \end{equation}
  In contrast to~\cref{E2ne0}, the term under $\sum_j\E_j$ has no constant contribution,
  but is of order $\q_\pa\R_j$.
  There is no cancellation for $q_\pa\to0$, and no need to use a series expansion
  for the accurate computation of~$f$.
\end{remark}

\begin{remark}[polyhedron with symmetry $C_i$]\label{R3inv}
  If a $(2\KK)$ has an inversion center at the origin~$\v{0}$
  (Schoenflies group $C_i$),
  then its face indices can be chosen such that $\Gamma_{k+\KK}=\v{0}-\Gamma_k$,
  and the face form factors have the symmetry
  $f(\q,\Gamma_{k+\KK})=f(-\q,\Gamma_k)=f(\q,\Gamma_k)^*$.
  The decomposition of $\q$ is invariant: $\q_{k+\KK\perp}=\q_{k\perp}$,
  and idem for $\q_{k\pa}$.
  In contrast, $\q_{k+\KK\times}=-\q_{k\times}$ because it involves $\n_{k+\KK}=-\n_k$.
  With this, the computation of $F(\q,\Pi)$ can be simplified as follows:

  The volume~\cref{E3Vol} can be computed as
  \begin{equation}\label{E3Vol_inv}
    \Vol(\Pi)
    = \frac{2}{3}\sum_{k=1}^\KK \Ar(\Gamma_k)\,r_{\perp k}.
  \end{equation}
  The analytic form factor~\cref{E3ne0} for $q\ne0$ can be computed as
  \begin{equation}\label{E3ne0inv}
      F(\q,\Pi)
      = \frac{1}{iq^2}\, \q^*\cdot \sum_{k=1}^\KK \n_k \,\tilde{f}(\q,\Gamma_k)
  \end{equation}
  with the antisymmetrized form factors of opposite faces
  \begin{equation}\label{E32finv}
      \tilde{f}(\q,\Gamma)
      \coloneqq f(\q,\Gamma)-f(-\q,\Gamma)
      = \frac{4}{iq_\pa^2}\qcross^*\cdot
      \sum_{j=1}^{J} \E_{j} \sinc(\q\E_{j}) \cos(\q\R_{j}),
  \end{equation}
  obtained from~\cref{E2fac} and~\cref{E2ne0}. If $\Gamma$ has symmetry $S_2$ they become
  \begin{equation}\label{E32finvS2}
      \tilde{f}(\q,\Gamma)
      \coloneqq f(\q,\Gamma)-f(-\q,\Gamma)
      = \frac{-8}{q_\pa^2}\qcross^*\cdot
      \sum_{j=1}^{\JJ} \E_{j} \sinc(\q\E_{j}) \sin(\qpa\R_{j}).
  \end{equation}
  If $q_{k\pa}$ is small, then the $\tilde{f}(\q,\Gamma_k)$
  should be computed from the series expansion
  \begin{equation}\label{E2series_inv}
    \tilde{f}(\q,\Gamma)
    = \sum_{n=0}^\infty i^n \tilde{\phi}_n(\q,\Gamma)
  \end{equation}
  with
  \begin{equation}
    \tilde{\phi}_n(\q,\Gamma)
    \coloneqq \e^{i\qrperp}f_n(\qpa) - \e^{-i\qrperp}f_n(-\qpa).
  \end{equation}
  From \cref{E2intcoeff} we see that
  the $f_n$ are even/odd functions of~$\qpa$ (but not of~$\q$) for even/odd~$n$.
  Thereby
  \begin{equation}\label{Eftilde_m}
    \tilde{\phi}_n(\q,\Gamma)
    = \left\{\begin{array}{ll}
       2i\sin(\qrperp)f_n(\qpa,\Gamma) &\text{~for even~$n$,}\\
       2\cos(\qrperp)f_n(\qpa,\Gamma) &\text{~for odd~$n$,}
      \end{array}\right.
  \end{equation}
  with $f_n$ computed using \cref{E2coeff}.
  In the remaining case of small $q$, the coefficients of the series~\cref{E3series}
  take the form
  \begin{equation}\label{E3ne0_inv}
    F_n(\q,\Pi)
    = \frac{1}{iq^2}\, \q^*\cdot \sum_{k=1}^\KK \n_k \:
       [f_{n+1}(\q,\Gamma)-f_{n+1}(-\q,\Gamma)],
  \end{equation}
  which offers hardly any advantage over the original form~\cref{E3coeff}.
\end{remark}

\begin{remark}[prism]\label{Rprism}
  For prisms, it is advisable to factorize~\cref{E3ffdef} from the onset.
  Let a prism~$\Pi$ have a base~$\Gamma_\pa$ that determines a normal~$\n$.
  Let the prism extend in normal direction from $-h/2$ to~$h/2$.
  Then for all~$\q$
  \begin{equation}
    F(\q,\Pi)
    = h\sinc(\q_\perp\n h/2)f(\qpa,\Gamma_\pa).
  \end{equation}
  Per \cref{Rsinc}, no series expansion is needed for $q_\perp\to0$.
\end{remark}

\section{Numeric tests}\label{Stests}

A form factor computation for generic polyhedra, based on all the above,
has been implemented as part of the GISAS simulation package BornAgain \cite{ffp:ba,ffp:ba:pub}.
The source code is available under the GNU General Public License.
All floating-point numbers, internally and externally, have double precision.
Tests were performed on standard PC's with amd64 architecture under Linux,
hence in IEEE arithmetics.

The code underwent extensive tests for internal consistency and for
 compatibility with previous computations.
Those reference computations were based on analytical expressions
obtained by three successive integrations in Cartesian coordinates;
they had been checked against the reference code IsGISAXS \cite{ReLL09,Laz08}.

The internal tests comprise the symmetry, specialization, and continuity tests.
All these tests were performed, as far as applicable, for a suite of particle shapes
(pyramidal frusta with 2-, 3-, 4-, and 6-fold symmetry,
cuboctahedron, truncated cube, regular dodecahedron and icosahedron),
for different wavevector directions $\v{\hat q}$
(along symmetry axis, slightly off such axis, or in completely unsymmetric directions;
purely real or with small imaginary parts),
and for a logarithmically wide range of magnitudes~$q$.
The main result of these tests is the worst-case relative deviation
between two form factor computations,
\begin{equation}
  \delta[F_1,F_2]
  \coloneqq \max_{\q} \frac{|F_1(\q)-F_2(\q)|}{|F_1(\q)+F_2(\q)|/2}
\end{equation}

Symmetry tests are performed for particle shapes
that are invariant under some rotation or reflection~$R$.
They yield $\delta[F(\q),F(R\q)]\lesssim5\cdot10^{-10}$.

In specialization tests,
special parameter sets are chosen for which two otherwise different figures $\Pi_1$, $\Pi_2$
coincide.
For instance, the rectangular base of a pyramid $\Pi_1$ is made a square,
so that $F(\q,\Pi_1)$ can be compared with the form factor of a square pyramid~$\Pi_2$.
Or the dihedral angle of a pyramid is set to $90^\circ$ so that the pyramid
coincides with a prism.
These tests yield $\delta[F(\q,\Pi_1),F(\q,\Pi_2)]\lesssim3\cdot10^{-10}$.

Finally, continuity tests search for possible discontinuities due to a change
in the computational method.
Using a precompiler switch,
additional source code lines are activated that tell the test program
whether the analytic expression or a series expansion was used to compute a form factor,
and in the latter case, at which order the summation was terminated.
For given direction~$\v{\hat q}$,
bisection in~$q$ is used to determine where such a change in computation method happens.
Then, the form factor slightly below and slightly above this threshold is determined.
With $\eta=8\cdot10^{-16}$ chosen as a few times the machine epsilon,
the continuity measure is $\delta_\text{cont}\coloneqq\delta[F(\q(1-\eta)),F(\q(1+\eta))]$.
According to \cref{R2final},
the switch between series expansion and analytic expression is determined by a parameter~$c$.
The optimum value of this parameter is determined empirically so that $\delta_\text{cont}$
is miminimized.

\begin{figure}[htbp]\label{F12sym}
  \centering
  \includegraphics[width=\linewidth]{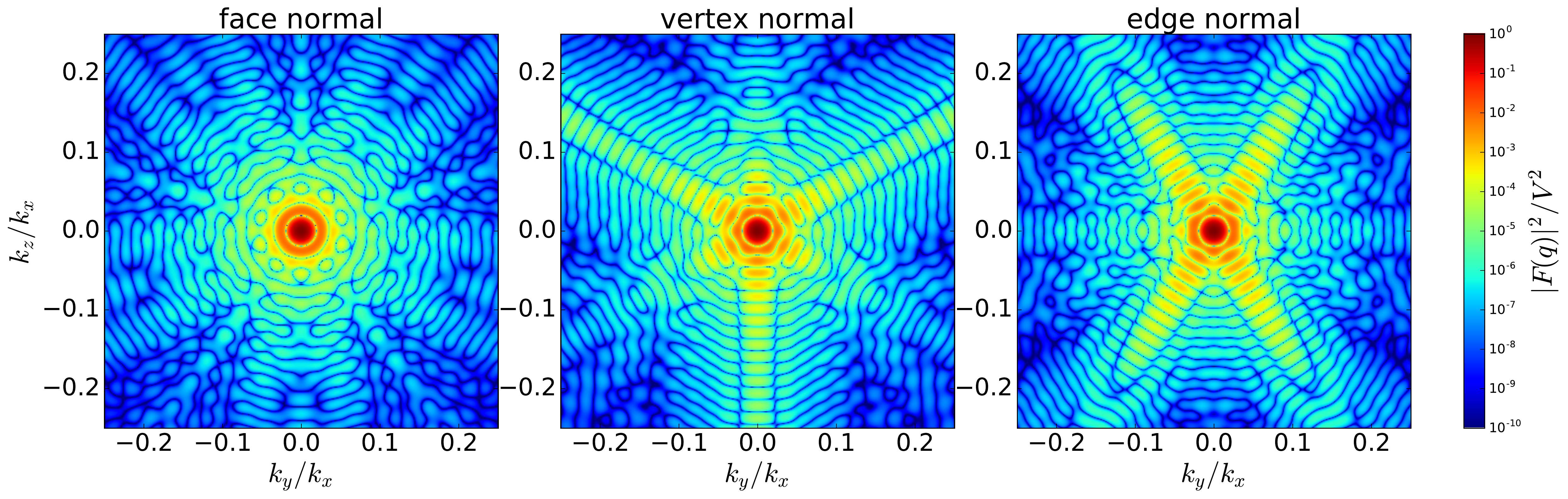}
  \caption{Normalized intensity $|F|^2/V^2$,
computed with $a=4.8$~nm,
for three orientations of high symmetry:
$x$ axis perpendicular to a polygonal face;
vertex on the $x$ axis;
edge in the $xy$ plane and perpendicular to the $x$ axis.}
\end{figure}

\begin{figure}[htbp]\label{F12asy}
  \centering
  \includegraphics[width=\linewidth]{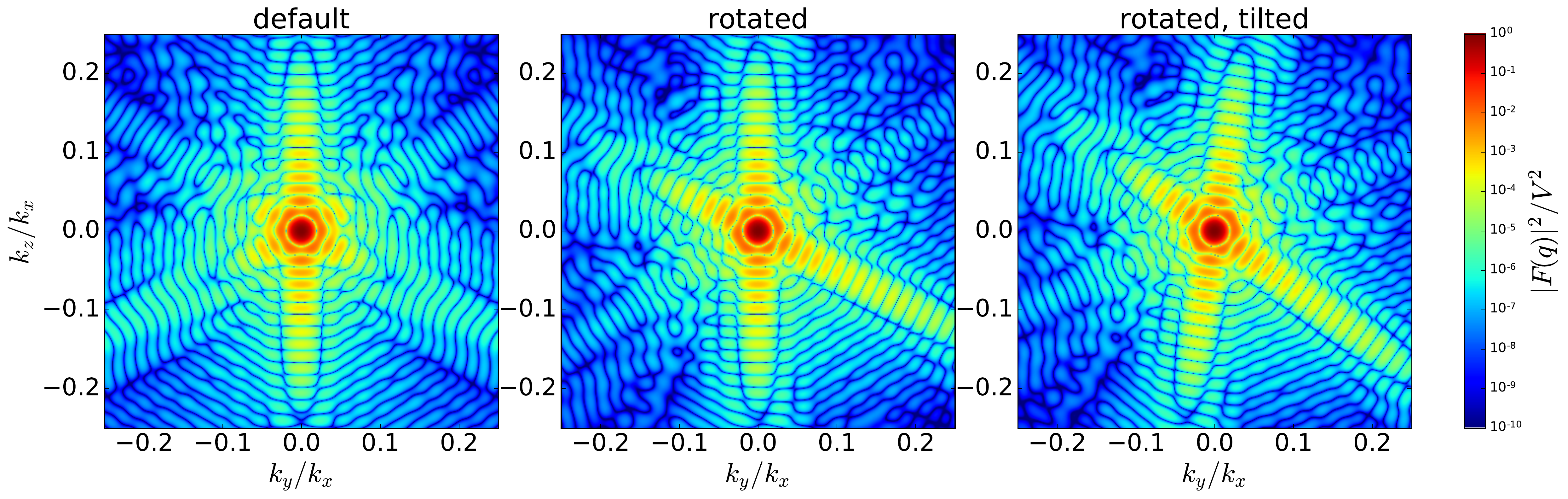}
  \caption{Normalized intensity $|F|^2/V^2$,
computed with $a=4.8$~nm,
for three orientations of decreasing symmetry:
base pentagon in $xy$ plane and pointing in $x$ direction;
rotated by $13^\circ$ around the $z$ axis;
ditto, and tilted by $9^\circ$ around the $x$ axis.}
\end{figure}



\section*{Acknowledgments}
I thank C\'eline Durniak, Walter Van Herck, and Gennady Pospelov
for reference code and for help with the test framework.

\end{document}